\documentclass{llncs}
\usepackage{algorithm}
\usepackage[noend]{algorithmic}
\usepackage{amsmath}
\usepackage{amssymb}
\usepackage{bbm}
\usepackage{bm}
\usepackage{color}
\usepackage{dsfont}
\usepackage{enumerate}
\usepackage{graphicx}
\usepackage[misc]{ifsym}
\usepackage{subfigure}
\usepackage{times}
\usepackage[usenames,dvipsnames]{xcolor}

\usepackage[bookmarks=false]{hyperref}
\hypersetup{
  pdftex,
  pdffitwindow=true,
  pdfstartview={FitH},
  pdfnewwindow=true,
  colorlinks,
  linktocpage=true,
  linkcolor=Green,
  urlcolor=Green,
  citecolor=Green
}
\usepackage[capitalize]{cleveref}

\newcommand{\commentout}[1]{}
\newcommand{\junk}[1]{}
\newcommand{\etal}{\emph{et al.}}

\newcommand{\naturalset}{\mathbb{N}}

\newcommand{\domain}[1]{\mathrm{dom}\left(#1\right)}

\newcommand{\parents}{\mathsf{pa}}

\newcommand{\abs}[1]{\left|#1\right|}

\newcommand{\EE}[1]{\mathbb{E} \left[#1\right]}

\newcommand{\I}[1]{\mathds{1} \! \left\{#1\right\}}

\newcommand{\set}[1]{\left\{#1\right\}}

\DeclareMathOperator*{\argmax}{arg\,max\,}

\DeclareMathOperator*{\median}{median\,}
\mathchardef\mhyphen="2D

\newcommand{\gmfactorsketch}{{\tt GMFactorSketch}}
\newcommand{\gmhash}{{\tt GMHash}}
\newcommand{\gmsketch}{{\tt GMSketch}}
\newcommand{\Pcm}{P_\textsc{cm}}
\newcommand{\Pml}{\tilde{P}}
\newcommand{\Pmlg}{\bar{P}}

\begin{document}

\title{Graphical Model Sketch}

\author{Branislav Kveton\inst{1} \Letter \and
Hung Bui\inst{2} \and
Mohammad Ghavamzadeh\inst{3} \and
Georgios Theocharous\inst{4} \and
S. Muthukrishnan\inst{5} \and
Siqi Sun\inst{6}}

\authorrunning{Kveton et al.}

\institute{Adobe Research, San Jose, CA \hfill \texttt{kveton@adobe.com} \and
Adobe Research, San Jose, CA \hfill \texttt{hubui@adobe.com} \and
Adobe Research, San Jose, CA \hfill \texttt{ghavamza@adobe.com} \and
Adobe Research, San Jose, CA \hfill \texttt{theochar@adobe.com} \and
Department of Computer Science, Rutgers, NJ \hfill \texttt{muthu@cs.rutgers.edu} \and
TTI, Chicago, IL \hfill \texttt{siqi.sun@ttic.edu}}

\maketitle

\begin{abstract}
Structured high-cardinality data arises in many domains, and poses a major challenge for both modeling and inference. Graphical models are a popular approach to modeling structured data but they are unsuitable for high-cardinality variables. The count-min (CM) sketch is a popular approach to estimating probabilities in high-cardinality data but it does not scale well beyond a few variables. In this work, we bring together the ideas of graphical models and count sketches; and propose and analyze several approaches to estimating probabilities in structured high-cardinality streams of data. The key idea of our approximations is to use the structure of a graphical model and approximately estimate its factors by ``sketches'', which hash high-cardinality variables using random projections. Our approximations are computationally efficient and their space complexity is independent of the cardinality of variables. Our error bounds are multiplicative and significantly improve upon those of the CM sketch, a state-of-the-art approach to estimating probabilities in streams. We evaluate our approximations on synthetic and real-world problems, and report an order of magnitude improvements over the CM sketch.
\end{abstract}


\section{Introduction}
\label{sec:introduction}

Structured high-cardinality data arises in numerous domains, and poses a major challenge for modeling and inference. A common goal in online advertising is to estimate the probability of events, such as page views, over multiple high-cardinality variables, such as the location of the user, the referring page, and the purchased product. A common goal in natural language processing is to estimate the probability of $n$-grams over a dictionary of $100\text{k}$ words. Graphical models \cite{jensen96introduction} are a popular approach to modeling multivariate data. However, when the cardinality of random variables is high, they are expensive to store and reason with. For instance, a graphical model over two variables with $M = 10^5$ values each may consume $M^2 = 10^{10}$ space.

A \emph{sketch} \cite{muthukrishnan05data} is a data structure that summarizes streams of data such that any two sketches of individual streams can be combined space efficiently into the sketch of the combined stream. Numerous problems can be solved efficiently by surprisingly simple sketches, such as estimating the frequency of values in streams \cite{misra82finding,charikar04finding,cormode05improved}, finding heavy hitters \cite{cormode05what}, estimating the number of unique values \cite{flajolet85probabilistic,flajolet07hyperloglog}, or even approximating low-rank matrices \cite{liberty13simple,woodruff14low}. In this work, we sketch a graphical model in a small space. Let $(x^{(t)})_{t = 1}^n$ be a stream of $n$ observations from some distribution $P$, where $x^{(t)} \in [M]^K$ is a $K$-dimensional vector and $P$ factors according to a known graphical model $\mathcal{G}$. Let $\Pmlg$ be the maximum-likelihood estimate (MLE) of $P$ from $(x^{(t)})_{t = 1}^n$ conditioned on $\mathcal{G}$. Then our goal is to approximate $\Pmlg$ with $\hat{P}$ such that $\hat{P}(x) \approx \Pmlg(x)$ for any $x \in [M]^K$ with at least $1 - \delta$ probability; in the space that does not depend on the cardinality $M$ of the variables in $\mathcal{G}$. In our motivating examples, $x$ is an $n$-gram or the feature vector associated with page views.

This paper makes three contributions. First, we propose and carefully analyze three natural approximations to the MLE in graphical models with high-cardinality variables. The key idea of our approximations is to leverage the structure of the graphical model $\mathcal{G}$ and approximately estimate its factors by ``sketches''. Therefore, we refer to our approximations as \emph{graphical model sketches}. Our best approximation, $\gmfactorsketch$, guarantees that $\hat{P}(x)$ is a constant-factor multiplicative approximation to $\Pmlg(x)$ for any $x$ with probability of at least $1 - \delta$ in $O(K^2 \log(K / \delta) \Delta^{-1}(x))$ space, where $K$ is the number of variables and $\Delta(x)$ measures the hardness of query $x$. The dependence on $\Delta(x)$ is generally unavoidable and we show this in \cref{sec:lower bound}. Second, we prove that $\gmfactorsketch$ yields better approximations than the count-min (CM) sketch \cite{cormode05improved}, a state-of-the-art approach to estimating the frequency of values in streams (\cref{sec:comparison}). Third, we evaluate our approximations on both synthetic and real-world problems. Our results show that $\gmfactorsketch$ outperforms the CM sketch and our other approximations, as measured by the error in estimating $\Pmlg$ at the same space.

Our work is related to Matusevych \etal~\cite{matusevych12hokusai}, who proposed several extensions of the CM sketch, one of which is $\gmfactorsketch$. This approximation is not analyzed and it is evaluated only on a graphical model with three variables. We present the first analysis of $\gmfactorsketch$, and prove that it is superior to other natural approximations and the CM sketch. We also evaluate $\gmfactorsketch$ on an order of magnitude larger problems than Matusevych \etal~\cite{matusevych12hokusai}. McGregor and Vu \cite{mcgregor15evaluating} proposed and analyzed a space-efficient streaming algorithm that tests if the stream of data is consistent with a graphical model. Several recent papers applied hashing to speeding up inference in graphical models \cite{ermon13taming,belle15hashing}. These papers do not focus on high-cardinality variables and are only loosely related to our work, because of using hashing in graphical models. We also note that the problem of representing conditional probabilities in graphical models efficiently has been studied extensively, as early as in Boutilier \etal~\cite{boutilier96contextspecific}. Our paper is different from this line of work because we do not assume any sparsity or symmetry in data; and our approximations are suitable for the streaming setting.

We denote $\set{1, \dots, K}$ by $[K]$. The cardinality of set $A$ is $\abs{A}$. We denote random variables by capital letters, such as $X$, and their values by small letters, such as $x$. We assume that $X = (X_1, \dots, X_K)$ is a $K$-dimensional variable; and we refer to its $k$-th component by $X_k$ and its value by $x_k$.


\section{Background}
\label{sec:background}

This section reviews the two main components of our solutions.

\subsection{Count-Min Sketch}
\label{sec:count-min sketch}

Let $(x^{(t)})_{t = 1}^n$ be a stream of $n$ observations from distribution $P$, where $x^{(t)} \in [M]^K$ is a $K$-dimensional vector. Suppose that we want to estimate:
\begin{align}
  \Pml(x) = \frac{1}{n} \sum_{t = 1}^n \I{x = x^{(t)}}\,,
  \label{eq:distribution}
\end{align}
the frequency of observing any $x$ in $(x^{(t)})_{t = 1}^n$. This problem can be solved in $O(M^K)$ space, by counting all unique values in $(x^{(t)})_{t = 1}^n$. This solution is impractical when $K$ and $M$ are large. Cormode and Muthukrishnan \cite{cormode05improved} proposed an approximate solution to this problem, the \emph{count-min (CM) sketch}, which estimates $\Pml(x)$ in the space independent of $M^K$. The sketch consists of $d$ hash tables with $m$ bins, $c \in \naturalset^{d \times m}$. The hash tables are initialized with zeros. At time $t$, they are updated with observation $x^{(t)}$ as:
\begin{align*}
  c(i, y) \gets c(i, y) + \I{y = h^i(x^{(t)})}
\end{align*}
for all $i \in [d]$ and $y \in [m]$, where $h^i: [M]^K \to [m]$ is the $i$-th \emph{hash function}. The hash functions are \emph{random} and \emph{pairwise-independent}. The frequency $\Pml(x)$ is estimated as:
\begin{align}
  \Pcm(x) = \frac{1}{n} \min\nolimits_{i \in [d]} c(i, h^i(x))\,.
  \label{eq:count-min sketch}
\end{align}
Cormode and Muthukrishnan \cite{cormode05improved} showed that $\Pcm(x)$ approximates $\Pml(x)$ for any $x \in [M]^K$, with at most $\varepsilon$ error and at least $1 - \delta$ probability, in $O((1 / \varepsilon) \log(1 / \delta))$ space. Note that the space is independent of $M^K$. We state this result more formally below.

\begin{theorem}
\label{thm:count-min sketch} Let $\Pml$ be the distribution in \eqref{eq:distribution} and $\Pcm$ be its CM sketch in \eqref{eq:count-min sketch}. Let $d = \log(1 / \delta)$ and $m = e / \varepsilon$. Then for any $x \in [M]^K$, $\Pml(x) \leq \Pcm(x) \leq \Pml(x) + \varepsilon$ with at least $1 - \delta$ probability. The space complexity of $\Pcm$ is $(e / \varepsilon) \log(1 / \delta)$.
\end{theorem}

\noindent The CM sketch is popular because high-quality approximations, with at most $\varepsilon$ error, can be computed in $O(1 / \varepsilon)$ space.\footnote{https://sites.google.com/site/countminsketch/} Other similar sketches, such as Charikar \etal~\cite{charikar04finding}, require $O(1 / \varepsilon^2)$ space.

\subsection{Bayesian Networks}
\label{sec:Bayesian networks}

Graphical models are a popular tool for modeling and reasoning with random variables \cite{koller09probabilistic}, and have many applications in computer vision \cite{murphy04using} and natural language processing \cite{lafferty01conditional}. In this work, we focus on Bayesian networks \cite{jensen96introduction}, which are directed graphical models.

A \emph{Bayesian network} is a probabilistic graphical model that represents conditional independencies of random variables by a directed graph. In this work, we define it as a pair $(\mathcal{G}, \theta)$, where $\mathcal{G}$ is a directed graph and $\theta$ are its parameters. The graph $\mathcal{G} = (V, E)$ is defined by its nodes $V = \set{X_1, \dots, X_K}$, one for each random variable, and edges $E$. For simplicity of exposition, we assume that $\mathcal{G}$ is a \emph{tree} and $X_1$ is its root. We relax this assumption in \cref{sec:model}. Under this assumption, each node $X_k$ for $k \geq 2$ has one parent and the probability of $x = (x_1, \dots, x_K)$ factors as:
\begin{align*}
  P(x) = P_1(x_1) \prod_{k = 2}^K P_k(x_k \mid x_{\parents(k)})\,,
\end{align*}
where $\parents(k)$ is the \emph{index of the parent variable} of $X_k$, and we use shorthands:
\begin{align*}
  P_k(i) = P(X_k = i)\,, \ \ 
  P_k(i, j) = P(X_k = i, X_{\parents(k)} = j)\,, \ \ 
  P_k(i \mid j) = \frac{P_k(i, j)}{P_{\parents(k)}(j)}\,.
\end{align*}
Let $\domain{X_k} = M$ for all $k \in [K]$. Then our graphical model is parameterized by $M$ \emph{prior probabilities} $P_1(i)$, for any $i \in [M]$; and $(K - 1) M^2$ \emph{conditional probabilities} $P_k(i \mid j)$, for any $k \in [K] - \set{1}$ and $i, j \in [M]$.

Let $(x^{(t)})_{t = 1}^n$ be $n$ observations of $X$. Then the \emph{maximum-likelihood estimate (MLE)} of $P$ conditioned on $\mathcal{G}$, $\bar{\theta} = \argmax_\theta P((x^{(t)})_{t = 1}^n \mid \theta, \mathcal{G})$, has a closed-form solution:
\begin{align}
  \Pmlg(x) = \Pmlg_1(x_1) \prod_{k = 2}^K \Pmlg_k(x_k \mid x_{\parents(k)})\,,
  \label{eq:factored distribution}
\end{align}
where we abbreviate $P(X = x \mid \bar{\theta}, \mathcal{G})$ as $\Pmlg(x)$, and define:
\begin{align*}
  \forall i \in [M]: \Pmlg_k(i) & = \frac{1}{n} \sum_{t = 1}^n \I{x^{(t)}_k = i}\,, \\
  \forall i, j \in [M]: \Pmlg_k(i, j) & = \frac{1}{n} \sum_{t = 1}^n \I{x^{(t)}_k = i, x^{(t)}_{\parents(k)} = j}\,, \\
  \forall i, j \in [M]: \Pmlg_k(i \mid j) & = \Pmlg_k(i, j) / \Pmlg_{\parents(k)}(j)\,.
\end{align*}


\section{Model}
\label{sec:model}

Let $(x^{(t)})_{t = 1}^n$ be a stream of $n$ observations from distribution $P$, where $x^{(t)} \in [M]^K$ is a $K$-dimensional vector. Our objective is to approximate $\Pmlg(x)$ in \eqref{eq:factored distribution}, the frequency of observing $x$ as given by the MLE of $P$ from $(x^{(t)})_{t = 1}^n$ conditioned on graphical model $\mathcal{G}$. This objective naturally generalizes that of the CM sketch in \eqref{eq:distribution}, which is the MLE of $P$ from $(x^{(t)})_{t = 1}^n$ without any assumptions on the structure of $P$. For simplicity of exposition, we assume that $\mathcal{G}$ is a tree (\cref{sec:Bayesian networks}). Under this assumption, $\Pmlg$ can be represented exactly in $O(K M^2)$ space. This is not feasible in our problems of interest, where typically $M \geq 10^4$.

The key idea in our solutions is to estimate a surrogate parameter $\hat{\theta}$. We estimate $\hat{\theta}$ on the same graphical model as $\bar{\theta}$. The difference is that $\hat{\theta}$ parameterizes a graphical model where each factor is represented by $O(m)$ hashing bins, where $m \ll M^2$. Our proposed models consume $O(K m)$ space, a significant reduction from $O(K M^2)$; and guarantee that $\hat{P}(x) \approx \Pmlg(x)$ for any $x \in [M]^K$ and observations $(x^{(t)})_{t = 1}^n$ up to time $n$, where we abbreviate $P(X = x \mid \hat{\theta}, \mathcal{G})$ as $\hat{P}(x)$. More precisely:
\begin{align}
  \Pmlg(x) \prod_{k = 1}^K [1 - \varepsilon_k] \leq
  \hat{P}(x) \leq
  \Pmlg(x) \prod_{k = 1}^K [1 + \varepsilon_k]
  \label{eq:objective}
\end{align}
for any $x \in [M]^K$ with at least $1 - \delta$ probability, where $\hat{P}$ is factored in the same way as $\Pmlg$. Each term $\varepsilon_k$ is $O(1 / m)$, where $m$ is the number of hashing bins. Therefore, the quality of our approximations improves as $m$ increases. More precisely, if $m$ is chosen such that $\varepsilon_k \leq 1 / K$ for all $k \in [K]$, we get:
\begin{align}
  [2 / (3 e)] \Pmlg(x) \leq \hat{P}(x) \leq e \Pmlg(x)
  \label{eq:e approximation}
\end{align}
for $K \geq 2$ since $\prod_{k = 1}^K (1 + \varepsilon_k) \leq (1 + 1 / K)^K \leq e$ for $K \geq 1$ and $\prod_{k = 1}^K (1 - \varepsilon_k) \geq (1 - 1 / K)^K \geq 2 / (3 e)$ for $K \geq 2$. Therefore, $\hat{P}(x)$ is a constant-factor multiplicative approximation to $\Pmlg(x)$. As in the CM sketch, we do not require that $\hat{P}(x)$ sum up to $1$.


\section{Summary of Main Results}
\label{sec:summary}

The main contribution of our work is that we propose and analyze three approaches to the MLE in graphical models with high-cardinality variables. Our first proposed algorithm, $\gmhash$ (\cref{sec:hashing}), approximates $\Pmlg(x)$ as the product of $K - 1$ conditionals and a prior, one for each variable in $\mathcal{G}$. Each conditional is estimated as a ratio of two hashing bins. $\gmhash$ guarantees \eqref{eq:e approximation} for any $x \in [M]^K$ with at least $1 - \delta$ probability in $O(K^3 \delta^{-1} \Delta^{-1}(x))$ space, where $\Delta(x)$ is a query-specific constant and the number of hashing bins is set as $m = \Omega(K^2 \delta^{-1})$. We discuss $\Delta(x)$ at the end of this section. Since $\delta$ is typically small, the dependence on $1 / \delta$ is undesirable.

Our second algorithm, $\gmsketch$ (\cref{sec:sketch}), approximates $\Pmlg(x)$ as the median of $d$ probabilities, each of which is estimated by $\gmhash$. $\gmsketch$ guarantees \eqref{eq:e approximation} for any $x \in [M]^K$ with at least $1 - \delta$ probability in $O(K^3 \log(1 / \delta) \Delta^{-1}(x))$ space, when we set $m = \Omega(K^2 \Delta^{-1}(x))$ and $d = \Omega(\log(1 / \delta))$. The main advantage over $\gmhash$ is that the space is $O(\log(1 / \delta))$ instead of $O(1 / \delta)$.

Our last algorithm, $\gmfactorsketch$ (\cref{sec:factor sketch}), approximates $\Pmlg(x)$ as the product of $K - 1$ conditionals and a prior, one for each variable. Each conditional is estimated as a ratio of two count-min sketches. $\gmfactorsketch$ guarantees \eqref{eq:e approximation} for any $x \in [M]^K$ with at least $1 - \delta$ probability in $O(K^2 \log(K / \delta) \Delta^{-1}(x))$ space, when we set $m = \Omega(K \Delta^{-1}(x))$ and $d = \Omega(\log(K / \delta))$. The key improvement over $\gmsketch$ is that the space is $O(K^2)$ instead of being $O(K^3)$. In summary, $\gmfactorsketch$ is the best of our proposed solutions. We demonstrate this empirically in \cref{sec:experiments}.

The query-specific constant $\Delta(x) = \min_{k \in [K] - \set{1}} \Pmlg_k(x_k, x_{\parents(k)})$ is the minimum probability that the values of any variable-parent pair in $x$ co-occur in $(x^{(t)})_{t = 1}^n$. This probability can be small and our algorithms are unsuitable for estimating $\Pmlg(x)$ in such cases. Note that this does not imply that $\Pmlg(x)$ cannot be small. Unfortunately, the dependence on $\Delta(x)$ is generally unavoidable and we show this in \cref{sec:lower bound}.

The assumption that $\mathcal{G}$ is a tree is only for simplicity of exposition. Our algorithms and their analysis generalize to the setting where $X_{\parents(k)}$ is a vector of parent variables and $x_{\parents(k)}$ are their values. The only change is in how the pair $(x_k, x_{\parents(k)})$ is hashed.


\section{Algorithms and Analysis}
\label{sec:algorithms and analysis}

All of our algorithms hash the values of each variable in graphical model $\mathcal{G}$, and each variable-parent pair, to $m$ bins up to $d$ times. We denote the $i$-th hash function of variable $X_k$ by $h^i_k$ and the associated hash table by $c_k(i, \cdot)$. This hash table approximates $n \Pmlg_k(\cdot)$. The $i$-th hash function of the variable-parent pair $(X_k, X_{\parents(k)})$ is also $h^i_k$, and the associated hash table is $\bar{c}_k(i, \cdot)$. This hash table approximates $n \Pmlg_k(\cdot, \cdot)$. Our algorithms differ in how the hash tables are aggregated.

We define the notion of a \emph{hash}, which is a tuple $h = (h_1, \dots, h_K)$ of $K$ \mbox{randomly} drawn hash functions $h_k: \naturalset \to [m]$, one for each variable in $\mathcal{G}$. We make the assumption that hashes are pairwise-independent. We say that hashes $h^i$ and $h^j$ are \emph{pairwise-independent} when $h^i_k$ and $h^j_k$ are pairwise-independent for all $k \in [K]$. These kinds of hash functions can be computed fast and stored in a very small space \cite{cormode05improved}.

\subsection{Algorithm $\gmhash$}
\label{sec:hashing}

\begin{algorithm}[t]
  \caption{$\gmhash$: Hashed conditionals and priors.}
  \label{alg:hashing}
  \begin{algorithmic}
    \STATE {\bf Input:} Point query $x = (x_1, \dots, x_K)$
    \STATE
    \STATE $\displaystyle \hat{P}_1(x_1) \gets \frac{c_1(h_1(x_1))}{n}$
    \FORALL{$k = 2, \dots, K$}
      \STATE $\displaystyle \hat{P}_k(x_k \mid x_{\parents(k)}) \gets
      \frac{\bar{c}_k(h_k(x_k + M (x_{\parents(k)} - 1)))}{c_{\parents(k)}(h_{\parents(k)}(x_{\parents(k)}))}$
    \ENDFOR
    \STATE $\displaystyle \hat{P}(x) \gets \hat{P}_1(x_1) \prod_{k = 2}^K \hat{P}_k(x_k \mid x_{\parents(k)})$
    \STATE \vspace{0.1in}
    \STATE {\bf Output:} Point answer $\hat{P}(x)$
  \end{algorithmic}
\end{algorithm}

The pseudocode of our first algorithm, $\gmhash$, is in \cref{alg:hashing}. It approximates $\Pmlg(x)$ as the product of $K - 1$ conditionals and a prior, one for each variable $X_k$. Each conditional is estimated as a ratio of two hashing bins:
\begin{align*}
  \hat{P}_k(x_k \mid x_{\parents(k)}) =
  \frac{\bar{c}_k(h_k(x_k + M (x_{\parents(k)} - 1)))}{c_{\parents(k)}(h_{\parents(k)}(x_{\parents(k)}))}\,,
\end{align*}
where $\bar{c}_k(h_k(x_k + M (x_{\parents(k)} - 1)))$ is the number of times that hash function $h_k$ maps $(x^{(t)}_k, x^{(t)}_{\parents(k)})$ to the same bin as $(x_k, x_{\parents(k)})$ in $n$ steps, and $c_k(h_k(x_k))$ is the number of times that $h_k$ maps $x^{(t)}_k$ to the same bin as $x_k$ in $n$ steps. Note that $(x_k, x_{\parents(k)})$ can be represented equivalently as $x_k + M (x_{\parents(k)} - 1)$. The prior $\Pmlg_1(x_1)$ is estimated as:
\begin{align*}
  \hat{P}_1(x_1) = \frac{1}{n} c_1(h_1(x_1))\,.
\end{align*}
At time $t$, the hash tables are updated as follows. Let $x^{(t)}$ be the observation. Then for all $k \in [K], y \in [m]$:
\begin{align*}
  c_k(y) & \gets c_k(y) + \I{y = h_k(x^{(t)}_k)}\,, \\
  \bar{c}_k(y) & \gets \bar{c}_k(y) + \I{y = h_k(x^{(t)}_k + M (x^{(t)}_{\parents(k)} - 1))}\,.
\end{align*}
This update takes $O(K)$ time.

$\gmhash$ maintains $2 K - 1$ hash tables with $m$ bins each, one for each variable and one for each variable-parent pair in $\mathcal{G}$. Therefore, it consumes $O(K m)$ space. Now we show that $\hat{P}$ is a good approximation of $\Pmlg$.

\begin{theorem}
\label{thm:hashing} Let $\hat{P}$ be the estimator from \cref{alg:hashing}. Let $h$ be a random hash and $m$ be the number of bins in each hash function. Then for any $x$:
\begin{align*}
  \Pmlg(x) \prod_{k = 1}^K (1 - \varepsilon_k) \leq
  \hat{P}(x) \leq
  \Pmlg(x) \prod_{k = 1}^K (1 + \varepsilon_k)
\end{align*}
holds with at least $1 - \delta$ probability, where:
\begin{align*}
  \varepsilon_1 = 2 K [\Pmlg_1(x_1) \delta m]^{-1}\,, \quad
  \forall k \in [K] - \set{1}: \varepsilon_k = 2 K [\Pmlg_k(x_k, x_{\parents(k)}) \delta m]^{-1}\,.
\end{align*}
\end{theorem}
\begin{proof}
The proof is in Appendix. The key idea is to show that the number of bins $m$ can be chosen such that:
\begin{align}
  |\hat{P}_k(x_k \mid x_{\parents(k)}) - \Pmlg_k(x_k \mid x_{\parents(k)})| > \varepsilon_k
  \label{eq:overestimate}
\end{align}
is not likely for any $k \in [K] - \set{1}$ and $\varepsilon_1, \dots, \varepsilon_K > 0$. In other words, we argue that our estimate of each conditional $\Pmlg_k(x_k \mid x_{\parents(k)})$ can be arbitrary precise. By Lemma 1 in Appendix, the necessary conditions for event \eqref{eq:overestimate} are:
\begin{align*}
  \frac{1}{n} c_{\parents(k)}(h_{\parents(k)}(x_{\parents(k)})) -
  \Pmlg_{\parents(k)}(x_{\parents(k)}) & > \varepsilon_k \alpha_k\,, \\
  \!\!\frac{1}{n} \bar{c}_k(h_k(x_k + M (x_{\parents(k)} - 1))) -
  \Pmlg_k(x_k, x_{\parents(k)}) & > \varepsilon_k \alpha_k\,,
\end{align*}
where $\alpha_k = \Pmlg_{\parents(k)}(x_{\parents(k)})$ is the frequency that $X_{\parents(k)} = x_{\parents(k)}$ in $(x^{(t)})_{t = 1}^n$. In short, event \eqref{eq:overestimate} can happen only if $\gmhash$ significantly overestimates either $\Pmlg_{\parents(k)}(x_{\parents(k)})$ or $\Pmlg_k(x_k, x_{\parents(k)})$. We bound the probability of these events using Markov's inequality (Lemma 2 in Appendix) and then get that none of the events in \eqref{eq:overestimate} happen with at least $1 - \delta$ probability when the number of hashing bins $m \geq \sum_{k = 1}^K (2 / (\varepsilon_k \alpha_k \delta))$. Finally, we choose appropriate $\varepsilon_1, \dots, \varepsilon_K$.
\end{proof}

\noindent \cref{thm:hashing} shows that $\hat{P}(x)$ is a multiplicative approximation to $\Pmlg(x)$. The approximation improves with the number of bins $m$ because all error terms $\varepsilon_k$ are $O(1 / m)$. The accuracy of the approximation depends on the frequency of interaction between the values in $x$. In particular, if $\Pmlg_k(x_k, x_{\parents(k)})$ is sufficiently large for all $k \in [K] - \set{1}$, the approximation is good even for small $m$. More precisely, under the assumptions that:
\begin{align*}
  m \geq 2 K^2 [\Pmlg_1(x_1) \delta]^{-1}\,, \quad
  \forall k \in [K] - \set{1}: m \geq 2 K^2 [\Pmlg_k(x_k, x_{\parents(k)}) \delta]^{-1}\,,
\end{align*}
all $\varepsilon_k \leq 1 / K$ and the bound in \cref{thm:hashing} reduces to \eqref{eq:e approximation} for $K \geq 2$.

\subsection{Algorithm $\gmsketch$}
\label{sec:sketch}

\begin{algorithm}[t]
  \caption{$\gmsketch$: Median of $d$ $\gmhash$ estimates.}
  \label{alg:sketch}
  \begin{algorithmic}
    \STATE {\bf Input:} Point query $x = (x_1, \dots, x_K)$
    \STATE
    \FORALL{$i = 1, \dots, d$}
      \STATE $\displaystyle \hat{P}^i_1(x_1) \gets \frac{c_1(i, h^i_1(x_1))}{n}$
      \FORALL{$k = 2, \dots, K$}
        \STATE $\displaystyle \hat{P}^i_k(x_k \mid x_{\parents(k)}) \gets
        \frac{\bar{c}_k(i, h^i_k(x_k + M (x_{\parents(k)} - 1)))}{c_{\parents(k)}(i, h^i_{\parents(k)}(x_{\parents(k)}))}$
      \ENDFOR
      \STATE $\displaystyle \hat{P}^i(x) \gets \hat{P}^i_1(x_1) \prod_{k = 2}^K \hat{P}^i_k(x_k \mid x_{\parents(k)})$
    \ENDFOR
    \STATE $\hat{P}(x) \gets \median_{i \in [d]} \hat{P}^i(x)$
    \STATE
    \STATE {\bf Output:} Point answer $\hat{P}(x)$
  \end{algorithmic}
\end{algorithm}

The pseudocode of our second algorithm, $\gmsketch$, is in \cref{alg:sketch}. The algorithm approximates $\Pmlg(x)$ as the median of $d$ probability estimates:
\begin{align*}
  \hat{P}(x) = \median\nolimits_{i \in [d]} \hat{P}^i(x)\,.
\end{align*}
Each $\hat{P}^i(x)$ is computed by one instance of $\gmhash$, which is associated with the hash $h^i = (h^i_1, \dots, h^i_K)$. At time $t$, the hash tables are updated as follows. Let $x^{(t)}$ be the observation. Then for all $k \in [K], i \in [d], y \in [m]$:
\begin{align}
  c_k(i, y) & \gets c_k(i, y) + \I{y = h^i_k(x^{(t)}_k)}\,, \label{eq:counter update} \\
  \bar{c}_k(i, y) & \gets \bar{c}_k(i, y) + \I{y = h^i_k(x^{(t)}_k + M (x^{(t)}_{\parents(k)} - 1))}\,. \nonumber
\end{align}
This update takes $O(K d)$ time. $\gmsketch$ maintains $d$ instances of $\gmhash$. Therefore, it consumes $O(K m d)$ space. Now we show that $\hat{P}$ is a good approximation of $\Pmlg$.

\begin{theorem}
\label{thm:sketch} Let $\hat{P}$ be the estimator from \cref{alg:sketch}. Let $h^1, \dots, h^d$ be $d$ random and pairwise-independent hashes, and $m$ be the number of bins in each hash function. Then for any $d \geq 8 \log(1 / \delta)$ and $x$:
\begin{align*}
  \Pmlg(x) \prod_{k = 1}^K (1 - \varepsilon_k) \leq
  \hat{P}(x) \leq
  \Pmlg(x) \prod_{k = 1}^K (1 + \varepsilon_k)
\end{align*}
holds with at least $1 - \delta$ probability, where $\varepsilon_k$ are defined in \cref{thm:hashing} for $\delta = 1 / 4$.
\end{theorem}
\begin{proof}
The proof is in Appendix. The key idea is the so-called median trick on $d$ estimates of $\gmhash$ in \cref{thm:hashing} for $\delta = 1 / 4$.
\end{proof}

\noindent Similarly to \cref{sec:hashing}, \cref{thm:sketch} shows that $\hat{P}(x)$ is a multiplicative approximation to $\Pmlg(x)$. The approximation improves with the number of bins $m$ and depends on the frequency of interaction between the values in $x$.

\subsection{Algorithm $\gmfactorsketch$}
\label{sec:factor sketch}

\begin{algorithm}[t]
  \caption{$\gmfactorsketch$: Count-min sketches of conditionals and priors.}
  \label{alg:factor sketch}
  \begin{algorithmic}
    \STATE {\bf Input:} Point query $x = (x_1, \dots, x_K)$
    \STATE
    \STATE // Count-min sketches for variables in $\mathcal{G}$
    \FORALL{$k = 1, \dots, K$}
      \FORALL{$i = 1, \dots, d$}
        \STATE $\displaystyle \hat{P}^i_k(x_k) \gets \frac{c_k(i, h^i_k(x_k))}{n}$
      \ENDFOR
      \STATE $\hat{P}_k(x_k) \gets \min_{i \in [d]} \hat{P}^i_k(x_k)$
    \ENDFOR
    \STATE
    \STATE // Count-min sketches for variable-parent pairs in $\mathcal{G}$
    \FORALL{$k = 2, \dots, K$}
      \FORALL{$i = 1, \dots, d$}
        \STATE $\displaystyle \hat{P}^i_k(x_k, x_{\parents(k)}) \gets
        \frac{\bar{c}_k(i, h^i_k(x_k + M (x_{\parents(k)} - 1)))}{n}$
      \ENDFOR
      \STATE $\hat{P}_k(x_k, x_{\parents(k)}) \gets \min_{i \in [d]} \hat{P}^i_k(x_k, x_{\parents(k)})$
    \ENDFOR
    \STATE
    \FORALL{$k = 2, \dots, K$}
      \STATE $\displaystyle \hat{P}_k(x_k \mid x_{\parents(k)}) \gets
      \frac{\hat{P}_k(x_k, x_{\parents(k)})}{\hat{P}_{\parents(k)}(x_{\parents(k)})}$
    \ENDFOR
    \STATE $\displaystyle \hat{P}(x) \gets \hat{P}_1(x_1) \prod_{k = 2}^K \hat{P}_k(x_k \mid x_{\parents(k)})$
    \STATE \vspace{0.1in}
    \STATE {\bf Output:} Point answer $\hat{P}(x)$
  \end{algorithmic}
\end{algorithm}

Our final algorithm, $\gmfactorsketch$, is in \cref{alg:factor sketch}. The algorithm approximates $\Pmlg(x)$ as the product of $K - 1$ conditionals and a prior, one for each variable $X_k$. Each conditional is estimated as a ratio of two CM sketches:
\begin{align*}
  \hat{P}_k(x_k \mid x_{\parents(k)}) =
  \frac{\hat{P}_k(x_k, x_{\parents(k)})}{\hat{P}_{\parents(k)}(x_{\parents(k)})}\,,
\end{align*}
where $\hat{P}_k(x_k, x_{\parents(k)})$ is the CM sketch of $\Pmlg_k(x_k, x_{\parents(k)})$ and $\hat{P}_k(x_k)$ is the CM sketch of $\Pmlg_k(x_k)$. The prior $\Pmlg_1(x_1)$ is approximated by its CM sketch $\hat{P}_1(x_1)$.

At time $t$, the hash tables are updated in the same way as in \eqref{eq:counter update}. This update takes $O(K d)$ time and $\gmfactorsketch$ consumes $O(K m d)$ space. Now we show that $\hat{P}$ is a good approximation of $\Pmlg$.

\begin{theorem}
\label{thm:factor sketch} Let $\hat{P}$ be the estimator from \cref{alg:factor sketch}. Let $h^1, \dots, h^d$ be $d$ random and pairwise-independent hashes, and $m$ be the number of bins in each hash function. Then for any $d \geq \log(2 K / \delta)$ and $x$:
\begin{align*}
  \Pmlg(x) \prod_{k = 1}^K (1 - \varepsilon_k) \leq
  \hat{P}(x) \leq
  \Pmlg(x) \prod_{k = 1}^K (1 + \varepsilon_k)
\end{align*}
holds with at least $1 - \delta$ probability, where:
\begin{align*}
  \varepsilon_1 = e [\Pmlg_1(x_1) m]^{-1}\,, \quad
  \forall k \in [K] - \set{1}: \varepsilon_k =  e [\Pmlg_k(x_k, x_{\parents(k)}) m]^{-1}\,.
\end{align*}
\end{theorem}
\begin{proof}
The proof is in Appendix. The main idea of the proof is similar to that of \cref{thm:hashing}. The key difference is that we prove that event \eqref{eq:overestimate} is unlikely for any $k \in [K] -\allowbreak \set{1}$ by bounding the probabilities of events:
\begin{align*}
  \hat{P}_{\parents(k)}(x_{\parents(k)}) - \Pmlg_{\parents(k)}(x_{\parents(k)}) & > \varepsilon_k \alpha_k\,, \\
  \hat{P}_k(x_k, x_{\parents(k)}) - \Pmlg_k(x_k, x_{\parents(k)}) & > \varepsilon_k \alpha_k\,,
\end{align*}
where $\hat{P}_k(x_k, x_{\parents(k)})$ is the CM sketch of $\Pmlg_k(x_k, x_{\parents(k)})$ and $\hat{P}_{\parents(k)}(x_{\parents(k)})$ is the CM sketch of $\Pmlg_{\parents(k)}(x_{\parents(k)})$.
\end{proof}

\noindent As in Sections \ref{sec:hashing} and \ref{sec:sketch}, \cref{thm:factor sketch} shows that $\hat{P}(x)$ is a multiplicative approximation to $\Pmlg(x)$. The approximation improves with the number of bins $m$ and depends on the frequency of interaction between the values in $x$.


\subsection{Lower Bound}
\label{sec:lower bound}

Our bounds depend on query-specific constants $\Pmlg_k(x_k, x_{\parents(k)})$, which can be small. We argue that this dependence is intrinsic. In particular, we show that there exists a family of distributions $\mathcal{C}$ such that any data structure that can summarize any $\Pmlg \in \mathcal{C}$ well must consume $\Omega(\Delta^{-1}(\mathcal{C}))$ space, where:
\begin{align*}
  \textstyle
  \Delta(\mathcal{C}) = \min_{\Pmlg \in \mathcal{C}, x \in [M]^K, k \in [K] - \set{1}: \Pmlg(x) > 0}
  \Pmlg_k(x_k,  x_{\parents(k)})\,.
\end{align*}
Our family of distributions $\mathcal{C}$ is defined on two dependent random variables, where $X_1$ is the parent and $X_2$ is its child. Let $m$ be an integer such that $m = 1 / \epsilon$ for some fixed $\epsilon \in [0, 1]$. Each model in $\mathcal{C}$ is defined as follows. The probability of any $m$ values of $X_1$ is $\epsilon$. The conditional of $X_2$ is defined as follows. When $\Pmlg_1(i) > 0$, the probability of any $m$ values of $X_2$ is $\epsilon$. When $\Pmlg_1(i) = 0$, the probability of all values of $X_2$ is $1 / M$. Note that each model induces a different distribution and that the number of the distributions is ${M \choose  m}^{m + 1}$, because there are ${M \choose  m}$ different priors $\Pmlg_1$ and ${M \choose  m}$ different conditionals $\Pmlg_2(\cdot \mid i)$, one for each $\Pmlg_1(i) > 0$. We also note that $\Delta(\mathcal{C}) = \epsilon^2$. The main result of this section is proved below.

\begin{theorem}
\label{thm:lower bound} Any data structure that can summarize any $\Pmlg \in \mathcal{C}$ as $\hat{P}$ such that $|\hat{P}(x) - \Pmlg(x)| < \epsilon^2 / 2$ for any $x \in [M]^K$ must consume $\Omega(\Delta^{-1}(\mathcal{C}))$ space.
\end{theorem}
\begin{proof}
Suppose that a data structure can summarize any $\Pmlg \in \mathcal{C}$ as $\hat{P}$ such that $|\hat{P}(x) - \Pmlg(x)| < \epsilon^2 / 2$ for any $x \in [M]^K$. Then the data structure must be able to distinguish between any two $\bar{P} \in \mathcal{C}$, since $\bar{P}(x) \in \set{0, \epsilon^2}$. At the minimum, such a data structure must be able to represent the index of any $\Pmlg \in \mathcal{C}$, which cannot be done in less than:
\begin{align*}
  \textstyle
  \log_2\left({M \choose  m}^{m + 1}\right) \geq
  \log_2\left(\left(M / m\right)^{m^2 + m}\right) \geq
  m^2 \log_2(M / m)
\end{align*}
bits because the number of distributions in $\mathcal{C}$ is ${M \choose  m}^{m + 1}$. Now note that $m^2 = 1 / \epsilon^2 = \Delta^{-1}(\mathcal{C})$.
\end{proof}

\noindent It is easy to verify that $\gmfactorsketch$ is such a data structure for $m = 5 e \Delta^{-1}(\mathcal{C})$ in \cref{thm:factor sketch}. In this setting, $\gmfactorsketch$ consumes $O(\log(1 / \delta) \Delta^{-1}(\mathcal{C}))$ space. The only major difference from \cref{thm:lower bound} is that $\gmfactorsketch$ makes a mistake with at most $\delta$ probability. Up to this factor, our analysis is order-optimal and we conclude that the dependence on the reciprocal of $\min_{k \in [K] - \set{1}} \Pmlg_k(x_k, x_{\parents(k)})$ cannot be avoided in general.


\section{Comparison with the Count-Min Sketch}
\label{sec:comparison}

In general, the error bounds in Theorems \ref{thm:count-min sketch} and \ref{thm:factor sketch} are not comparable, because $\Pml$ in \eqref{eq:distribution} is a different estimator from $\Pmlg$ in \eqref{eq:factored distribution}. To compare the bounds, we make the assumption that $(x^{(t)})_{t = 1}^n$ is a stream of $n$ observations such that $\Pmlg = \Pml$. This holds, for instance, when $n \to \infty$, because both $\Pmlg$ and $\Pml$ are consistent estimators of $P$. In the rest of this section, and without loss of generality, we assume that $\Pmlg = \Pml = P$.

In this section, we construct a class of graphical models where $\gmfactorsketch$ has a tighter error bound than the CM sketch. This class contains naive Bayes models with $K + 1$ variables:
\begin{align}
  P(x) = P_1(x_1) \prod_{k = 2}^{K + 1} P_k(x_k \mid x_1)\,.
  \label{eq:NB}
\end{align}
Variable $X_1$ is binary. For any $k \in [K + 1] - \set{1}$, variable $X_k$ takes values from $[M]$. For simplicity of exposition, we assume that the prior is $P_1(1) = P_1(2) = 0.5$. We fix $x$ and define $C_k = P_k(x_k \mid x_1)$ for any $k \in [K + 1] - \set{1}$.

Suppose that $\gmfactorsketch$ represents $P_1$ exactly, and therefore $\hat{P}_1 = P_1$. Then by \cref{thm:factor sketch}, for any $x$ with at least $1 - \delta$ probability:
\begin{align}
  \hat{P}(x) \leq \frac{1}{2} \left[\prod_{k = 2}^{K + 1} C_k\right]
  \left[\prod_{k = 2}^{K + 1} \left(1 + \frac{2 e}{C_k m}\right)\right]\,,
  \label{eq:NB factor sketch}
\end{align}
where $m$ is the number of hashing bins in $\gmfactorsketch$. Since $\hat{P}_1 = P_1$, we can omit $1 + \varepsilon_1$ from \cref{thm:factor sketch}. This approximation consumes, up to logarithmic factors in $K$, $2 K m \log(1 / \delta)$ space. The CM sketch (\cref{sec:count-min sketch}) guarantees that:
\begin{align}
  \Pcm(x) \leq
  \frac{1}{2} \left[\prod_{k = 2}^{K + 1} C_k\right] + \frac{e}{m'} =
  \frac{1}{2} \left[\prod_{k = 2}^{K + 1} C_k\right]
  \left(1 + \frac{2 e}{m'} \left[\prod_{k = 2}^{K + 1} \frac{1}{C_k}\right]\right)
  \label{eq:NB count-min sketch}
\end{align}
for any $x$ with at least $1 - \delta$ probability, where $m'$ is the number of hashing bins in the CM sketch. This approximation consumes $m' \log(1 / \delta)$ space.

We want to show that the upper bound in \eqref{eq:NB factor sketch} is tighter than that in \eqref{eq:NB count-min sketch} for any reasonable $m$. Since $\gmfactorsketch$ maintains $2 K$ times more hash tables than the CM sketch, we increase the number of bins in the CM sketch to $m' = 2 K m$, and get the following upper bound:
\begin{align}
  \Pcm(x)
  & \leq \frac{1}{2} \left[\prod_{k = 2}^{K + 1} C_k\right]
  \left(1 + \frac{e}{K m} \left[\prod_{k = 2}^{K + 1} \frac{1}{C_k}\right]\right)\,.
  \label{eq:NB count-min sketch 2}
\end{align}
Now both $\gmfactorsketch$ and the CM sketch consume the same space, and their error bounds are functions of $m$.

Roughly speaking, the bound in \eqref{eq:NB factor sketch} seems to be tighter than that in \eqref{eq:NB count-min sketch 2} because it contains $K$ potentially large values $1 / C_k$, each of which can be offset by a potentially small $1 / m$. On the other hand, all values $1 / C_k$ in \eqref{eq:NB count-min sketch 2} are offset only by a single $1 / m$. Now we prove this claim formally. Before we start, note that both upper bounds in \eqref{eq:NB factor sketch} and \eqref{eq:NB count-min sketch 2} contain $\frac{1}{2} \left[\prod_{k = 2}^{K + 1} C_k\right]$. Therefore, we can divide both bounds by this constant and get that the upper bound in \eqref{eq:NB factor sketch} is tighter than that in \eqref{eq:NB count-min sketch 2} when:
\begin{align}
  1 + \frac{e}{K m} \left[\prod_{k = 2}^{K + 1} \frac{1}{C_k}\right] >
  \prod_{k = 2}^{K + 1} \left(1 + \frac{2 e}{C_k m}\right)\,.
  \label{eq:NB event 1}
\end{align}
Now we rewrite each $(1 + 2 e / (C_k m))$ on the right-hand side as $(1 / C_k) (C_k + 2 e / m)$ and multiply both sides by $\prod_{k = 2}^{K + 1} C_k$. Then we omit $\prod_{k = 2}^{K + 1} C_k$ from the left-hand side and get that event \eqref{eq:NB event 1} happens when:
\begin{align}
  \frac{e}{K m} > \prod_{k = 2}^{K + 1} \left(C_k + \frac{2 e}{m}\right)\,.
  \label{eq:NB event 2}
\end{align}
If $C_k$ is close to one for all $k \in [K + 1] - \set{1}$, the right-hand side of \eqref{eq:NB event 2} is at least one and we get that $m$ should be smaller than $e / K$. This result is impractical since $K$ is usually much larger than $e$ and we require that $m \geq 1$. To make progress, we restrict our analysis to a class of $x$. In particular, let $C_k \leq 1 / 2$ for all $k \in [K + 1] - \set{1}$. Then we can bound the right-hand side of \eqref{eq:NB event 2} from above as:
\begin{align*}
  \prod_{k = 2}^{K + 1} \left(C_k + \frac{2 e}{m}\right) \leq
  \left(\frac{1}{2}\right)^K \left(1 + \frac{4 e}{m}\right)^K \leq
  e \left(\frac{1}{2}\right)^K
\end{align*}
for $m \geq 4 e K$. This assumption on $m$ is not particularly strong, since \cref{thm:factor sketch} says that we get good multiplicative approximations to $\Pmlg(x)$ only if $m = \Omega(K)$. Now we apply the above upper bound to inequality \eqref{eq:NB event 2} and rearrange it as $2^K / K > m$. Since $2^K / K$ is exponential in $K$, we get that the bound in \eqref{eq:NB factor sketch} is tighter than that in \eqref{eq:NB count-min sketch 2} for a wide range of $m$ and any $x$ where $C_k \leq 1 / 2$ for all $k \in [K + 1] - \set{1}$. Our result is summarized below.

\begin{theorem}
\label{thm:comparison} Let $P$ be the distribution in \eqref{eq:NB} and $x$ be such that $P_k(x_k \mid x_1) \leq 1 / 2$ for all $k \in [K + 1] - \set{1}$. Let $m \geq 4 e K$ and $m' = 2 K m$. Then for any $m < 2^K / K$, the error bound of $\gmfactorsketch$ is tighter than that of the CM sketch at the same space. More precisely:
\begin{align*}
  P(x) \prod_{k = 2}^{K + 1} (1 + \varepsilon_k) \leq
  P(x) + \frac{e}{m'}\,,
\end{align*}
where $\varepsilon_k$ are defined in \cref{thm:factor sketch}.
\end{theorem}

The above result is quite practical. Suppose that $K = 32$. Then our upper bound is tighter for any $m$ such that:
\begin{align*}
  4 e K < 348 \leq m \leq
  2^{27} = 2^{32} / 32 = 2^K / K\,.
\end{align*}
By the pidgeonhole principle, \cref{thm:comparison} guarantees improvements in at least $2 (M - 1)^K$ points $x$ in any distribution in \eqref{eq:NB}. We can bound the fraction of these points from below as:
\begin{align*}
  \frac{2 (M - 1)^K}{2 M^K} =
  \exp[K \log(M - 1) - K \log M] \geq
  \exp\left[- \frac{K}{M - 1}\right] \geq
  1 - \frac{K}{M - 1}\,.
\end{align*}
In our motivating examples, $M \approx 10^5$ and $K \approx 100$. In this setting, the error bound of $\gmfactorsketch$ is tighter than that of the CM sketch in at least $99.9\%$ of $x$, for any naive Bayes model in \eqref{eq:NB}.


\section{Experiments}
\label{sec:experiments}

In this section, we compare our algorithms (\cref{sec:algorithms and analysis}) and the CM sketch on the synthetic problem in \cref{sec:comparison}, and also on a real-world problem in online advertising.

\subsection{Synthetic Problem}
\label{sec:synthetic problem}

We experiment with the naive Bayes model in \eqref{eq:NB}, where $P_1(1) = P_1(2) = 0.5$; and:
\begin{align*}
  \forall i \in [N]: P_k(i \mid 1) & = 1 / N\,, \quad
  & \forall i \in [M] - [N]: P_k(i \mid 1) & = 0\,, \\
  \forall i \in [N]: P_k(i \mid 2) & = 0\,, \quad
  & \forall i \in [M] - [N]: P_k(i \mid 2) & = 1 / (M - N)
\end{align*}
for any $k \in [K + 1] - \set{1}$ and $N \ll M$. The model defines the following distribution over $x = (x_1, \dots, x_K)$: when $x_1 = 1$, $P(x) = 0.5 N^{- K}$ and we refer to the example $x$ as \emph{heavy}; and when $x_1 = 2$, $P(x) = 0.5 (M - N)^{- K}$ and we refer to the example $x$ as \emph{light}. The heavy examples are much more probable when $N \ll M$. We set $M = 2^{16}$.

All compared algorithms are trained on $1\text{M}$ i.i.d. examples from distribution $P$ and tested on $500\text{k}$ i.i.d. heavy examples from $P$. We report the fraction of imprecise estimates of $P$ as a function of space. The estimate of $P(x)$ is \emph{precise} when $(1 / e) P(x) \leq\allowbreak \hat{P}(x) \leq e P(x)$. When the sample size $n$ is large, both $\Pmlg \to P$ and $\Pml \to P$, and this is a fair way of comparing our methods to the CM sketch. We choose $d = 5$. We observe similar trends for other values of $d$. All results are averaged over $20$ runs.

\subsection{Easy Synthetic Problem}
\label{sec:easy synthetic problem}

We choose $K = 4$ and $N = 8$, and then $P(x) = 2^{-13}$ for all heavy $x$. In this problem, the CM sketch can approximate $P(x)$ within a multiplicative factor of $e$ for any heavy $x$ in about $2^{13}$ space. This space is small, and therefore this problem is \emph{easy for the CM sketch}.

Our results are reported in Figure \ref{fig:results}a. We observe that all of our algorithms outperform the CM sketch. In particular, note that $\Pcm$ approximates $P$ well for any heavy $x$ in about $2^{15}$ space. Our algorithms achieve the same quality of the approximation in at most $2^{13}$ space. $\gmfactorsketch$ consumes $2^{10}$ space, which is almost two orders of magnitude less than the CM sketch.

\subsection{Hard Synthetic Problem}
\label{sec:hard synthetic problem}

\begin{figure*}[t]
  \centering
  \includegraphics[width=2.32in, bb=2.8in 4.5in 5.7in 6.5in]{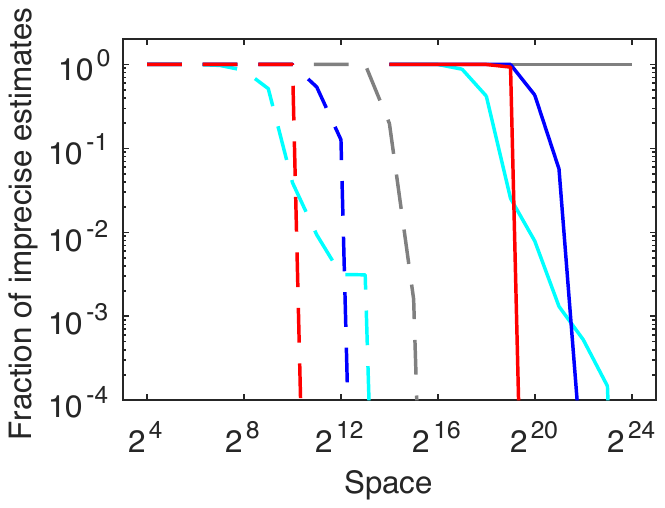}
  \includegraphics[width=2.32in, bb=2.8in 4.5in 5.7in 6.5in]{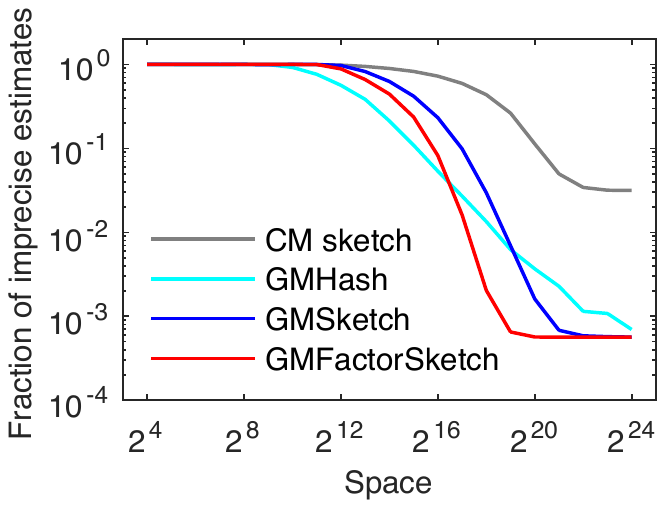} \\
  \vspace{0in} \hspace{0.2in} (a) \hspace{2.15in} (b) \hspace{0.05in}
  \caption{\textbf{a}. Evaluation of the CM sketch, $\gmhash$, $\gmsketch$, and $\gmfactorsketch$ on the easy problem
   in \cref{sec:easy synthetic problem} (dashed lines) and the hard problem in \cref{sec:hard synthetic problem} (solid lines). \textbf{b}. Evaluation on the real-world problem in \cref{sec:real-world problem}.}
  \label{fig:results}
\end{figure*}

We set $K = 32$ and $N = 64$, and then $P(x) = 2^{-193}$ for all heavy $x$. In this problem, the CM sketch can approximate $P(x)$ within a multiplicative factor of $e$ for any heavy $x$ in about $2^{193}$ space. This space is unrealistically large, and therefore this problem is \emph{hard for the CM sketch}.

Our results are reported in Figure \ref{fig:results}a and we observe three major trends. First, the CM sketch performs poorly. Second, as in \cref{sec:easy synthetic problem}, our algorithms outperform the CM sketch. Finally, when the fraction of imprecise estimates is small, our algorithms perform as suggested by our theory. $\gmhash$ is inferior to $\gmsketch$, which is further inferior to $\gmfactorsketch$.

\subsection{Real-World Problem}
\label{sec:real-world problem}

We also evaluate our algorithms on a real-world problem where the goal is to estimate the probability of a page view. We experiment with two months of data of a medium-sized customer of \emph{Adobe Marketing Cloud}\footnote{http://www.adobe.com/marketing-cloud.html}. This is $65\text{M}$ page views, each of which is described by six variables: \textsc{Country}, \textsc{City}, \textsc{Page Name}, \textsc{Starting Page Name}, \textsc{Campaign}, and \textsc{Browser}. Variable \textsc{Page Name} takes on more than $42\text{k}$ values and has the highest cardinality. We approximate the distribution $P$ over our variables by a naive Bayes model, where the class variable is $X_1 = \textsc{Country}$. Since the behavior of users is often driven by their locations, this approximation is quite reasonable.

All compared algorithms are trained on $1\text{M}$ i.i.d. examples from distribution $P$ and tested on all heavy examples in this sample. We say that the example $x$ is \emph{heavy} when $P(x) > 10^{- 6}$. The rest of the setup is identical to that in \cref{sec:synthetic problem}.

Our results are reported in Figure \ref{fig:results}b. We observe the same trends as in \cref{sec:hard synthetic problem}. The CM sketch performs poorly, and our methods outperform it at the same space for any space from $2^{13}$ to $2^{24}$. Also note that none of the compared methods achieve zero mistakes. This is because our sample size $n$ is not large enough to approximate $P$ well in all heavy $x$. Even if $\hat{P} = \Pmlg$, our methods would still make mistakes.


\section{Conclusions}
\label{sec:conclusions}

Structured high-cardinality data arises in many domains. Probability distributions over such data cannot be estimated easily with guarantees by either graphical models \cite{jensen96introduction}, a popular approach to reasoning with structured data; or count sketches \cite{muthukrishnan05data}, a common approach to approximating probabilities in high-cardinality streams of data. We bring together the ideas of graphical models and sketches, and propose three approximations to the MLE in graphical models with high-cardinality variables. We analyze them and prove that our best approximation, $\gmfactorsketch$, outperforms the CM sketch on a class of naive Bayes models. We validate these findings empirically.

The MLE is a common approach to estimating the parameters of graphical models \cite{jensen96introduction}. We propose, analyze, and empirically evaluate multiple space-efficient approximations to this procedure with high-cardinality variables. In this work, we focus solely on the problem of estimating $\Pmlg(x)$, the probability at a single point $x$. However, note that our models are constructed from Bayesian networks, which can answer $P(Y = y)$ for any subset of variables $Y$ with values $y$. We do not analyze such inference queries and leave this for future work.

Our work is the first formal investigation of approximations on the intersection of graphical models and sketches. One of our key results is that $\gmfactorsketch$ yields a constant-factor multiplicative approximation to $\Pmlg(x)$ for any $x$ with probability of at least $1 - \delta$ in $O(K^2 \log(K / \delta) \Delta^{-1}(x))$ space, where $K$ is the number of variables and $\Delta(x)$ reflects the hardness of query $x$. This result is encouraging because the space is only quadratic in $K$ and logarithmic in $1 / \delta$. The space also depends on constant $\Delta(x)$, which can be small. This constant is intrinsic (\cref{sec:lower bound}); and this indicates that the problem of approximating $\Pmlg(x)$ well, for any $\Pmlg$ and $x$, is intrinsically hard.

\clearpage

\bibliographystyle{plain}
\bibliography{References}


\clearpage
\onecolumn
\appendix

\section{Proofs of Main Theorems}
\label{sec:proofs}

\subsection{Proof of \cref{thm:hashing}}

First, we prove a supplementary claim that the number of bins $m$ can be set such that:
\begin{align}
  [\Pmlg_1(x_1) - \varepsilon_1] \prod_{k = 2}^K [\Pmlg_k(x_k \mid x_{\parents(k)}) - \varepsilon_k]
  & \leq \hat{P}(x) \label{eq:hashing raw} \\
  & \leq [\Pmlg_1(x_1) + \varepsilon_1] \prod_{k = 2}^K [\Pmlg_k(x_k \mid x_{\parents(k)}) + \varepsilon_k]
  \nonumber
\end{align}
holds with probability of at least $1 - \delta$ for any $\varepsilon_1, \dots, \varepsilon_K > 0$. Then we choose appropriate $\varepsilon_1, \dots, \varepsilon_K$. To prove that \eqref{eq:hashing raw} holds, it suffices to show that inequalities:
\begin{align}
  |\hat{P}_1(x_1) - \Pmlg_1(x_1)| & \leq \varepsilon_1\,,
  \label{eq:hashing good prior} \\
  \forall k \in [K] - \set{1}: |\hat{P}_k(x_k \mid x_{\parents(k)}) - \Pmlg_k(x_k \mid x_{\parents(k)})| & \leq \varepsilon_k
  \label{eq:hashing good conditional}
\end{align}
hold jointly with probability of at least $1 - \delta$.

Clearly $\hat{P}_1(x_1) - \Pmlg_1(x_1) \geq 0$. Therefore, the probability that \eqref{eq:hashing good prior} does not hold is bounded by \cref{lem:small collision} as:
\begin{align}
  P(|\hat{P}_1(x_1) - \Pmlg_1(x_1)| > \varepsilon_1) =
  P(\hat{P}_1(x_1) - \Pmlg_1(x_1) > \varepsilon_1) <
  \frac{1}{m \varepsilon_1}\,.
  \label{eq:hashing prior bound}
\end{align}
Now we fix $k \in [K] - \set{1}$ and bound the probability that \eqref{eq:hashing good conditional} does not hold:
\begin{align*}
  & P(|\hat{P}_k(x_k \mid x_{\parents(k)}) - \Pmlg_k(x_k \mid x_{\parents(k)})| > \varepsilon_k) = \\
  & \quad P\left(\abs{\frac{\bar{c}_k(h_k(x_k + M (x_{\parents(k)} - 1)))}
  {c_{\parents(k)}(h_{\parents(k)}(x_{\parents(k)}))} -
  \frac{\sum_{t = 1}^n \I{x^{(t)}_k = x_k, x^{(t)}_{\parents(k)} = x_{\parents(k)}}}
  {\sum_{t = 1}^n \I{x^{(t)}_{\parents(k)} = x_{\parents(k)}}}}
  > \varepsilon_k\right)\,.
\end{align*}
By \cref{lem:ratio decomposition}, the necessary conditions for $|\hat{P}_k(x_k \mid x_{\parents(k)}) - \Pmlg_k(x_k \mid x_{\parents(k)})| > \varepsilon_k$ are:
\begin{align*}
  \frac{1}{n} c_{\parents(k)}(h_{\parents(k)}(x_{\parents(k)})) -
  \frac{1}{n} \sum_{t = 1}^n \I{x^{(t)}_{\parents(k)} = x_{\parents(k)}} & > \varepsilon_k \alpha_k\,, \\
  \frac{1}{n} \bar{c}_k(h_k(x_k + M (x_{\parents(k)} - 1))) -
  \frac{1}{n} \sum_{t = 1}^n \I{x^{(t)}_k = x_k, x^{(t)}_{\parents(k)} = x_{\parents(k)}} & > \varepsilon_k \alpha_k\,,
\end{align*}
where $\alpha_k = \Pmlg_{\parents(k)}(x_{\parents(k)})$. The first event happens when the denominator of $\hat{P}_k(x_k \mid x_{\parents(k)})$ increases significantly when compared to the denominator of $\Pmlg_k(x_k \mid x_{\parents(k)})$. The second event happens when the numerator increases significantly.

Now we show that the above events are unlikely. The probability of the first event can be bounded by \cref{lem:small collision} as:
\begin{align}
  & P\left(\frac{1}{n} c_{\parents(k)}(h_{\parents(k)}(x_{\parents(k)})) -
  \frac{1}{n} \sum_{t = 1}^n \I{x^{(t)}_{\parents(k)} = x_{\parents(k)}} > \varepsilon_k \alpha_k\right)
  \nonumber \\
  & \qquad < \frac{1}{m \varepsilon_k \alpha_k}
  \label{eq:hashing den bound}
\end{align}
for $X = X_{\parents(k)}$, $h = h_{\parents(k)}$, and $\varepsilon = \varepsilon_k \alpha_k$. The probability of the second event can be bounded by \cref{lem:small collision} as:
\begin{align}
  & P\left(\frac{1}{n} \bar{c}_k(h_k(x_k + M (x_{\parents(k)} - 1))) -
  \frac{1}{n} \sum_{t = 1}^n \I{x^{(t)}_k = x_k, x^{(t)}_{\parents(k)} = x_{\parents(k)}} > \varepsilon_k \alpha_k\right)
  \nonumber \\
  & \qquad < \frac{1}{m \varepsilon_k \alpha_k}
  \label{eq:hashing num bound}
\end{align}
for $X = X_k + M (X_{\parents(k)} - 1)$, $h = h_k$, and $\varepsilon = \varepsilon_k \alpha_k$. Now we chain \eqref{eq:hashing prior bound}, \eqref{eq:hashing den bound}, and \eqref{eq:hashing num bound}; and have by the union that at least one inequality in \eqref{eq:hashing good prior} and \eqref{eq:hashing good conditional} is violated with probability of at most:
\begin{align*}
  \frac{1}{m \varepsilon_1} + \sum_{k = 2}^K \frac{2}{m \varepsilon_k \alpha_k} <
  \frac{1}{m} \sum_{k = 1}^K \frac{2}{\varepsilon_k \alpha_k}\,,
\end{align*}
where $\alpha_1 = 1$. This probability is bounded by $\delta$ for $\displaystyle m \geq \sum_{k = 1}^K \frac{2}{\varepsilon_k \alpha_k \delta}$. This concludes the proof of \eqref{eq:hashing raw}.

Now we choose appropriate $\varepsilon_1, \dots, \varepsilon_K$. In particular, let $\varepsilon_k = 2 K / (\alpha_k \delta m)$ for all $k \in [K]$. Note that this setting is valid for any $m \geq 1$ since:
\begin{align*}
  m \geq
  \sum_{k = 1}^K \frac{2}{\varepsilon_k \alpha_k \delta} =
  \sum_{k = 1}^K \frac{m}{K} = m\,.
\end{align*}
Under this assumption, the upper bound in \eqref{eq:hashing raw} can be written as:
\begin{align*}
  \hat{P}(x)
  & \leq [\Pmlg_1(x_1) + \varepsilon_1] \prod_{k = 2}^K [\Pmlg_k(x_k \mid x_{\parents(k)}) + \varepsilon_k] \\
  & = \left[\Pmlg_1(x_1) + \frac{2 K}{\alpha_1 \delta m}\right]
  \prod_{k = 2}^K \left[\Pmlg_k(x_k \mid x_{\parents(k)}) + \frac{2 K}{\alpha_k \delta m}\right] \\
  & = \left[\Pmlg_1(x_1) \prod_{k = 2}^K \Pmlg_k(x_k \mid x_{\parents(k)})\right]
  \left[1 + \frac{2 K}{\Pmlg_1(x_1) \delta m}\right]
  \prod_{k = 2}^K \left[1 + \frac{2 K}{\Pmlg_k(x_k, x_{\parents(k)}) \delta m}\right]\,.
\end{align*}
Along the same lines, the lower bound in \eqref{eq:hashing raw} can be written as:
\begin{align*}
  \hat{P}(x)
  & \geq [\Pmlg_1(x_1) - \varepsilon_1] \prod_{k = 2}^K [\Pmlg_k(x_k \mid x_{\parents(k)}) - \varepsilon_k] \\
  & = \left[\Pmlg_1(x_1) - \frac{2 K}{\alpha_1 \delta m}\right]
  \prod_{k = 2}^K \left[\Pmlg_k(x_k \mid x_{\parents(k)}) - \frac{2 K}{\alpha_k \delta m}\right] \\
  & = \left[\Pmlg_1(x_1) \prod_{k = 2}^K \Pmlg_k(x_k \mid x_{\parents(k)})\right]
  \left[1 - \frac{2 K}{\Pmlg_1(x_1) \delta m}\right]
  \prod_{k = 2}^K \left[1 - \frac{2 K}{\Pmlg_k(x_k, x_{\parents(k)}) \delta m}\right]\,.
\end{align*}
This concludes our proof.

\subsection{Proof of \cref{thm:sketch}}

Algorithm $\gmsketch$ estimates the probability as a median of $d$ probabilities:
\begin{align*}
  \hat{P}(x) = \median_{i \in [d]} \hat{P}^i(x)\,,
\end{align*}
each of which is estimated by a random instance of $\gmhash$. We bound the probability that $\hat{P}(x)$ is a good approximation of $\Pmlg(x)$:
\begin{align*}
  \Pmlg(x) \prod_{k = 1}^K (1 - \varepsilon_k) \leq
  \hat{P}(x) \leq
  \Pmlg(x) \prod_{k = 1}^K (1 + \varepsilon_k)\,,
\end{align*}
where $\varepsilon_k$ are defined in \cref{thm:hashing}, using the so-called median trick. Let:
\begin{align*}
  Z_i = \I{\Pmlg(x) \prod_{k = 1}^K (1 - \varepsilon_k) \leq \hat{P}^i(x) \leq \Pmlg(x) \prod_{k = 1}^K (1 + \varepsilon_k)}
\end{align*}
indicate the event that $\hat{P}^i(x)$ approximates $\Pmlg(x)$ well. In addition, let $\bar{Z} = \frac{1}{d} \sum_{i = 1}^d Z_i$ and $\EE{\bar{Z}} \geq 1 / 2$, where the expectation is with respect to random hashes $h^1, \dots, h^d$. Then by Hoeffding's inequality:
\begin{align*}
  P(\EE{\bar{Z}} - \bar{Z} > \EE{\bar{Z}} - 1 / 2) < \exp[- 2 (\EE{\bar{Z}} - 1 / 2)^2 d]\,,
\end{align*}
where $\EE{\bar{Z}} - \bar{Z} > \EE{\bar{Z}} - 1 / 2$ is the event that $\hat{P}(x)$ is not a good approximation of $\Pmlg(x)$. By setting $\delta = 1 / 4$ in \cref{thm:hashing}, we get that $\EE{\bar{Z}} \geq 3 / 4$ and therefore:
\begin{align*}
  P(\EE{\bar{Z}} - \bar{Z} > \EE{\bar{Z}} - 1 / 2) <
  \exp[- 2 (3 / 4 - 1 / 2)^2 d] =
  \exp[- d / 8]\,.
\end{align*}
Now we select $d \geq 8 \log(1 / \delta)$ and get that $\hat{P}(x)$ is a not a good approximation of $\Pmlg(x)$ with probability of at most $\delta$.

\subsection{Proof of \cref{thm:factor sketch}}

The key idea of this proof is similar to that of \cref{thm:hashing}. First, we prove a supplementary claim that the number of bins $m$ can be chosen such that:
\begin{align}
  [\Pmlg_1(x_1) - \varepsilon_1] \prod_{k = 2}^K [\Pmlg_k(x_k \mid x_{\parents(k)}) - \varepsilon_k]
  & \leq \hat{P}(x) \label{eq:factor sketch raw} \\
  & \leq [\Pmlg_1(x_1) + \varepsilon_1] \prod_{k = 2}^K [\Pmlg_k(x_k \mid x_{\parents(k)}) + \varepsilon_k]
  \nonumber
\end{align}
holds with probability of at least $1 - \delta$ for any $\varepsilon_1, \dots, \varepsilon_K > 0$. Then we choose appropriate $\varepsilon_1, \dots, \varepsilon_K$. To prove that \eqref{eq:factor sketch raw} holds, it suffices to show that inequalities:
\begin{align*}
  |\hat{P}_1(x_1) - \Pmlg_1(x_1)| & \leq \varepsilon_1\,, \\
  \forall k \in [K] - \set{1}: |\hat{P}_k(x_k \mid x_{\parents(k)}) - \Pmlg_k(x_k \mid x_{\parents(k)})| & \leq \varepsilon_k
\end{align*}
hold jointly with probability of at least $1 - \delta$. By \cref{lem:ratio decomposition} and the union bound, this is equivalent to showing that each of the following inequalities:
\begin{align*}
  \hat{P}_1(x_1) - \Pmlg_1(x_1)
  & \leq \varepsilon_1 \alpha_1\,, \\
  \forall k \in [K] - \set{1}: \hat{P}_{\parents(k)}(x_{\parents(k)}) - \Pmlg_{\parents(k)}(x_{\parents(k)})
  & \leq \varepsilon_k \alpha_k\,, \\
  \forall k \in [K] - \set{1}: \hat{P}_k(x_k, x_{\parents(k)}) - \Pmlg_k(x_k, x_{\parents(k)})
  & \leq \varepsilon_k \alpha_k
\end{align*}
is violated with probability of at most $\delta / (2 K)$, where $\alpha_1 = 1$ and $\alpha_k = \Pmlg_{\parents(k)}(x_{\parents(k)})$ for any $k \in [K] - \set{1}$. Now note that each $\hat{P}$ is the CM sketch of the corresponding $\Pmlg$. So, by Theorem 1 of Cormode and Muthukrishnan \cite{cormode05improved}, each of the above inequalities is violated with at most  $\delta / (2 K)$ probability when the number of hash functions satisfies $d \geq \log(2 K / \delta)$ and the number of hashing bins $m$ satisfies:
\begin{align*}
  m & \geq \frac{e}{\varepsilon_1 \alpha_1}\,, \\
  \forall k \in [K] - \set{1}: m & \geq \frac{e}{\varepsilon_k \alpha_k}\,.
\end{align*}
To satisfy the above inequalities, we select appropriate $\varepsilon_1, \dots, \varepsilon_K$. Let $\varepsilon_k = e / (\alpha_k m)$ for all $k \in [K]$. This setting is valid for any $m \geq 1$ and $k \in [K]$ since:
\begin{align*}
  m \geq
  \frac{e}{\varepsilon_k \alpha_k} =
  m\,.
\end{align*}
Under this assumption, the upper bound in \eqref{eq:factor sketch raw} can be written as:
\begin{align*}
  \hat{P}(x)
  & \leq [\Pmlg_1(x_1) + \varepsilon_1] \prod_{k = 2}^K [\Pmlg_k(x_k \mid x_{\parents(k)}) + \varepsilon_k] \\
  & = \left[\Pmlg_1(x_1) + \frac{e}{\alpha_1 m}\right]
  \prod_{k = 2}^K \left[\Pmlg_k(x_k \mid x_{\parents(k)}) + \frac{e}{\alpha_k m}\right] \\
  & = \left[\Pmlg_1(x_1) \prod_{k = 2}^K \Pmlg_k(x_k \mid x_{\parents(k)})\right]
  \left[1 + \frac{e}{\Pmlg_1(x_1) m}\right]
  \prod_{k = 2}^K \left[1 + \frac{e}{\Pmlg_k(x_k, x_{\parents(k)}) m}\right]\,.
\end{align*}
Along the same lines, the lower bound in \eqref{eq:factor sketch raw} can be written as:
\begin{align*}
  \hat{P}(x)
  & \geq [\Pmlg_1(x_1) - \varepsilon_1] \prod_{k = 2}^K [\Pmlg_k(x_k \mid x_{\parents(k)}) - \varepsilon_k] \\
  & = \left[\Pmlg_1(x_1) - \frac{e}{\alpha_1 m}\right]
  \prod_{k = 2}^K \left[\Pmlg_k(x_k \mid x_{\parents(k)}) - \frac{e}{\alpha_k m}\right] \\
  & = \left[\Pmlg_1(x_1) \prod_{k = 2}^K \Pmlg_k(x_k \mid x_{\parents(k)})\right]
  \left[1 - \frac{e}{\Pmlg_1(x_1) m}\right]
  \prod_{k = 2}^K \left[1 - \frac{e}{\Pmlg_k(x_k, x_{\parents(k)}) m}\right]\,.
\end{align*}
This concludes our proof.

\section{Technical Lemmas}
\label{sec:lemmas}

\begin{lemma}
\label{lem:ratio decomposition} Let:
\begin{align*}
  \abs{\frac{u_h}{v_h} - \frac{u}{v}} > \varepsilon
\end{align*}
for any $u_h \geq u$, $v_h \geq v$, $v \geq u$, and $v \geq \alpha n$. Then either $v_h - v > \varepsilon \alpha n$ or $u_h - u > \varepsilon \alpha n$.
\end{lemma}
\begin{proof}
The proof is by contradiction. First, note that $\displaystyle \abs{\frac{u_h}{v_h} - \frac{u}{v}} > \varepsilon$ implies that either:
\begin{align*}
  \frac{u_h}{v_h} - \frac{u}{v} > \varepsilon \qquad \text{or} \qquad
  \frac{u}{v} - \frac{u_h}{v_h} > \varepsilon\,.
\end{align*}
Now we argue that $u_h / v_h - u / v > \varepsilon$ implies $u_h - u > \varepsilon \alpha n$. Suppose that this is not true. Then the opposite must be true, $u_h / v_h - u / v > \varepsilon$ and $u_h - u \leq \varepsilon \alpha n$. We derive contradiction by bounding $\varepsilon$ from above as:
\begin{align*}
  \varepsilon <
  \frac{u_h}{v_h} - \frac{u}{v} =
  \underbrace{\frac{v}{v_h}}_{\leq 1} \frac{u_h}{v} - \frac{u}{v} \leq
  \frac{u_h - u}{v} \leq
  \frac{u_h - u}{\alpha n}\,.
\end{align*}
Now we argue that $u / v - u_h / v_h > \varepsilon$ implies $v_h - v > \varepsilon \alpha n$. Suppose that this is not true. Then the opposite must be true, $u / v - u_h / v_h > \varepsilon$ and $v_h - v \leq \varepsilon \alpha n$. We derive contradiction by bounding $\varepsilon$ from above as:
\begin{align*}
  \varepsilon <
  \frac{u}{v} - \frac{u_h}{v_h} =
  \frac{u}{v} - \underbrace{\frac{u_h}{u}}_{\geq 1} \frac{u}{v_h} \leq
  \underbrace{\frac{u}{v}}_{\leq 1} \frac{v_h - v}{v_h} \leq
  \frac{v_h - v}{\alpha n}\,.
\end{align*}
The last step follows from $v_h \geq v \geq \alpha n$. This concludes our proof.
\end{proof}

\begin{lemma}
\label{lem:small collision} Let $X$ be a discrete random variable on $\naturalset$ and $(x^{(t)})_{t = 1}^n$ be its $n$ observations. Let $h: \naturalset \to [m]$ be any random hash function. Then for any $x \in \naturalset$, $m \geq 1$, and $\varepsilon \in (0, 1)$:
\begin{align*}
  P\left(\frac{1}{n} \sum_{t = 1}^n \I{h(x^{(t)}) = h(x)} -
  \frac{1}{n} \sum_{t = 1}^n \I{x^{(t)} = x} > \varepsilon\right) < \frac{1}{m \varepsilon}\,,
\end{align*}
where the randomness is with respect to $h$.
\end{lemma}
\begin{proof}
Clearly:
\begin{align*}
  \frac{1}{n} \sum_{t = 1}^n \I{h(x^{(t)}) = h(x)} - \frac{1}{n} \sum_{t = 1}^n \I{x^{(t)} = x} \geq 0
\end{align*}
because $x^{(t)} = x$ implies that $h(x^{(t)}) = h(x)$ for any $h: \naturalset \to [m]$. Therefore, we can apply Markov's inequality and get:
\begin{align*}
  & P\left(\frac{1}{n} \sum_{t = 1}^n \I{h(x^{(t)}) = h(x)} - \frac{1}{n} \sum_{t = 1}^n \I{x^{(t)} = x} > \varepsilon\right) \\
  & \quad < \frac{1}{\varepsilon n} \EE{\sum_{t = 1}^n \I{h(x^{(t)}) = h(x)} - \sum_{t = 1}^n \I{x^{(t)} = x}} \\
  & \quad = \frac{1}{\varepsilon n} \sum_{t = 1}^n \EE{\I{h(x^{(t)}) = h(x), \ x^{(t)} \neq x}}\,,
\end{align*}
where the last equality is by the linearity of expectation. Since $h$ is random, the probability that $h(x^{(t)}) = h(x)$ when $x^{(t)} \neq x$ is $1 / m$. Therefore:
\begin{align*}
  \EE{\I{h(x^{(t)}) = h(x), \ x^{(t)} \neq x}} \leq 1 / m
\end{align*}
and we conclude that:
\begin{align*}
  P\left(\frac{1}{n} \sum_{t = 1}^n \I{h(x^{(t)}) = h(x)} - \frac{1}{n} \sum_{t = 1}^n \I{x^{(t)} = x} > \varepsilon\right) <
  \frac{1}{\varepsilon m}\,.
\end{align*}
\end{proof}

\end{document}